\documentclass[conference,10pt]{IEEEtran}
\ifCLASSINFOpdf
\else
\fi

\usepackage{eulervm}

\usepackage{graphicx}
\usepackage{array}
\usepackage{amsmath}
\usepackage{amssymb}
\usepackage{amsthm}
\usepackage{mathrsfs}
\usepackage{bbm}
\usepackage{bbold}
\usepackage{xfrac}
\usepackage{algorithmic}
\usepackage[linesnumbered,ruled]{algorithm2e}
\usepackage{upgreek}
\usepackage{multirow}
\usepackage{arydshln}
\usepackage{subfigure}
\usepackage{cite}
\usepackage{bm}
\usepackage{epstopdf}
\usepackage{mathtools}
\usepackage{CJKfntef}
\usepackage[all,pdf]{xy}

\let\emph\relax 

\begin{document}
%
\title{Optimizing Power and User Association for Energy Saving in Load-Coupled Cooperative LTE}


\author{\IEEEauthorblockN{Lei You$^1$, Lei Lei$^1$, and Di Yuan$^{1,2}$} 
\IEEEauthorblockA{\small$^1$Department of Science and Technology,
Link\"{o}ping University, Sweden\\}
\IEEEauthorblockA{\small$^2$Institute for Systems Research, University of Maryland, College Park, MD, 20740, USA\\}
\small\texttt{\{lei.you, lei.lei, di.yuan\}@liu.se, diyuan@umd.edu}
}


%


\maketitle

\begin{abstract} 
We consider an energy minimization problem for cooperative LTE networks. 
To reduce energy consumption, we investigate how to jointly optimize the transmit power and the association between cells and user equipments (UEs), by taking into consideration joint transmission (JT), one of the coordinated multipoint (CoMP) techniques. 
We formulate the optimization problem mathematically. 
For solving the problem, a dynamic power allocation algorithm that adjusts the transmit power of all cells, and an algorithm for optimizing the cell-UE association, are proposed. 
The two algorithms are iteratively used in an algorithmic framework to enhance the energy performance.
Numerically, the proposed algorithms can lead to lower energy consumption than the optimal energy setting in the non-JT case. 
In comparison to fixed power allocation in JT, the proposed dynamic power allocation algorithm is able to significantly reduce the energy consumption.
\end{abstract}


%
\IEEEpeerreviewmaketitle

\theoremstyle{plain}
\newtheorem{definition}{Definition}
\newtheorem{theorem}{Theorem}
\newtheorem{lemma}{Lemma}
\newtheorem{proposition}{Proposition}
\newtheorem{postulation}{Postulation}
\newtheorem{property}{Property}
\newtheorem{corollary}{Corollary}

\section{Introduction}
Improving energy efficiency is one of the urgent tasks in the future cellular systems, due to both economic and environmental reasons \cite{Ge:2016}.
The numerous user equipments (UEs) with high data traffic demand and the mass deployment of base stations can lead to high energy consumption in cellular networks \cite{Cavalcante:2014jd}. 
Considering the existence of inter-cell interference, how to satisfy all UEs' data demand with low energy consumption is challenging. 

To deliver data demand for the associated UEs in a cell, an amount of time-frequency resource will be consumed for data transmission.
The transmit power of a cell can affect the usage of time-frequency resource in the cell. 
The proportion of the consumed time-frequency resource for transmission in a cell is defined as cell's load. 
To characterize the influence between inter-cell interference and cells' load, a so called ``load-coupling" model has been proposed in \cite{IViering:2009tq,Ge:2014,Siomina:eq,Fehske:2012iw}. 
By adopting this model, several energy minimization problems have been investigated in \cite{Ho:2015hw,Anonymous:SLcg4Bln}, where the authors proved that, given users' demand, the total energy consumption can be reduced by increasing the load of all cells. However, the result is limited to non-cooperative cellular systems. In other words, no UE in the cellular network can be served by multiple cells simultaneously. 
This cooperation technique is called joint transmission (JT). As one of the coordinated multipoint (CoMP) techiniques, JT is able to enhance the efficiency of the time-frequency resource usage, via exploiting the potential interference signal reuse for data transmission \cite{Porcello:2014wq}.

In this paper, with JT, we investigate the energy consumption in networks by optimizing power allocation and cell-UE association.
We present the following contributions. 
First, we formulate the considered energy consumption problem mathematically.
Second, we provide theoretical analysis for the problem solving.
We show that when JT is taken into consideration, the full load optimality conclusion in \cite{Ho:2015hw,Anonymous:SLcg4Bln} may not hold in general. 
Third, considering the high computational complexity in solving the optimization problem, we propose two algorithms, for optimizing the power allocation and the cell-UE association, respectively. The two proposed algorithms can be jointly used to improve the energy performance in polynomial time.
In the proposed algorithm, based on our theoretical analysis, we systematically scale down the transmit power and optimize the cell-UE association for JT. 
Finally,
we numerically illustrate that the proposed algorithm is capable of improving the performance of network energy consumption, compared with the optimal energy setting \cite{Ho:2015hw} in the non-JT case. 
Moreover, we show that the proposed power allocation algorithm outperforms the fixed transmit power schemes in JT case, on the performance of energy consumption.

\textit{Notations:} We denote a (tall) vector by a bold lower case letter, say $\bm{x}$. We denote $\bm{x}>\bm{0}$ and $\bm{x}>\bm{1}$  if $x_i>0$ and $x_i>1$, respectively, for all $i$; similarly for the inequality $<$. We denote $\bm{x}\geq\bm{x'}$ if there exist at least one $i$ with $x_i>x'_i$, and for other $k\neq i$ we have $x_k\geq x'_k$; similarly for the less-equality $\leq$.

\section{System Model}

\subsection{Network Model}
Denote the set of all cells by $\mathcal{I}$. Denote the set of all UEs by $\mathcal{J}$. Let $n=|\mathcal{I}|$ and $m=|\mathcal{J}|$. Each UE is served at least by one cell with non-zero demand. 
Given that $\mathcal{I}_j$ denotes the set of cells currently serving UE $j$, and $\mathcal{J}_i$ the set of UEs currently served by cell $i$, respectively, the cell-UE association, is then decided. Note that $i\in\mathcal{I}_j\Leftrightarrow j\in\mathcal{J}_i$. This means, if $\mathcal{I}_j$ is fixed, then is $\mathcal{J}_i$, and vise versa. For clarity, we use both of them to indicate the cell-UE association, to keep the coherence in the context.

\subsection{Load Coupling}
We introduce the load coupling model in this subsection. 
For convenience, we fix $\mathcal{I}_j$ (and $\mathcal{J}_i$) in the expression of the signal-to-interference-and-noise-ratio (SINR) and the cell load. 

We define the load of any cell $i$ in Eq.~(\ref{eq:load}). The bitrate demand of UE $j$ is represented by $d_j~(d_j>0)$. The network-wide demand is represented by the vector $\bm{d}=[d_1,d_2,\ldots,d_m]$. The SINR of UE $j$ is denoted by $\gamma_j$. The network-wide SINR is denoted by $\bm{\gamma}=[\gamma_1,\gamma_2,\ldots,\gamma_m]$. We use the term ``resource unit (RU)'' to refer to one or more than one resource blocks in orthogonal-frequency-division multiple access (OFDMA). Without loss of generality, an RU is imposed to be the minimal unit for resource allocation. The bandwidth per RU is denoted by $B$. In the denominator in Eq.~(\ref{eq:sinr}), $B\log_2(1+\gamma_j)$ computes the achievable bitrate per RU. We assume there are $M$ RUs in total, such that $MB\log_2(1+\gamma_j)$ is the total achievable bitrate for UE $j$. Denote $y_j=d_j/(MB\log_2(1+\gamma_j))$, and $y_j$ represents the proportion of RUs required for UE $j$ to satisfy $d_j$. The summation of $y_j$ for all $j\in\mathcal{J}_i$, is called the load of cell $i$, represented by $x_i$. Denote the network-wide load by the vector $\bm{x}=[x_1,x_2,\ldots,x_n]$. As we can see in Eq.~(\ref{eq:load}), $x_i$ is a function of $\bm{\gamma}$, $\bm{d}$ and the association between $i$ and its serving UEs, aka. $\mathcal{J}_i$. We denote this function by $f_i(\bm{\gamma},\bm{d},\mathcal{J}_i)$.
\begin{equation}
x_i\coloneqq{f}_i(\bm{\gamma},\bm{d},\mathcal{J}_i)=\sum_{j\in\mathcal{J}_i}\frac{d_j}{MB\log_2\left(1+\gamma_j\right)}
\label{eq:load}
\end{equation}
\begin{equation}
\gamma_j\coloneqq{h}_j(\bm{x},\bm{p},\mathcal{I}_j)=\frac{\sum_{i\in \mathcal{I}_j}p_ig_{ij}}{\sum_{k\in \mathcal{I} \backslash \mathcal{I}_j}p_kg_{kj}x_k+\sigma^2}
\label{eq:sinr}
\end{equation}

Eq.~(\ref{eq:sinr}), shows the SINR of any UE $j$. Notation $p_i$ is the transmit power of cell $i$ per RU. The network-wide power is denoted by the vector $\bm{p}=[p_1,p_2,\ldots,p_n]$. The channel gain from cell $i$ to UE $j$ is denoted by $g_{ij}$. In the numerator, the transmit power received at UE $j$ from all its serving cells $\mathcal{I}_j$ is computed by $\sum_{i\in\mathcal{I}_j}p_ig_{ij}$. In the denominator, $\sum_{k\in\mathcal{I}\backslash\mathcal{I}_k}p_kg_{kj}x_k$ computes the interference power received at UE $j$ from cell $k$. In an extreme case that none of the RUs in cell $k$ is occupied for transmission, we have $x_k=0$ and the term $p_kg_{kj}x_k=0$. This means that UE $j$ does not receive interference on any RU of cell $k$, no matter which RUs $j$ is now occupied in its serving cells. On the contrary, if all RUs in cell $k$ are occupied for transmission (meaning that $k$ is in full load and $x_k=1$), then UE $j$ always receives interference from cell $k$, no matter which RUs are used by $j$ on its serving cells. In this case, the interference from cell $k$ to UE $j$ is computed by $p_kg_{kj}$. Generally, for any RU in cell $i$, the value of $x_k$ reflects the likelihood the cell $i$ receives the interference from cell $k$ on this RU. Note that the SINR $\gamma_j$ is a function of the network-wide power $\bm{p}$, the load $\bm{x}$, and the association between $j$ and its serving cells $\mathcal{I}_j$. We denote this function by $h_j(\bm{x},\bm{p},\mathcal{I}_j)$.

As we can see by the discussion above, for any UE $j$, the load of its interfering cells $k\in\mathcal{I}\backslash\mathcal{I}_j$ impacts the SINR $\gamma_j$, further causing an influence on the load of $j$'s serving cells $i\in\mathcal{I}_j$. This characteristics that the usage of RUs on different cells are mutually influenced, is called load coupling. Eq~(\ref{eq:load_coupling}) shows this relationship in a network-wide perspective.
\begin{equation}
\textbf{Load Coupling: }\left\{
\begin{array}{l}
\bm{\gamma}=\bm{h}(\bm{x},\bm{p},\hat{\mathcal{I}}) \\
\bm{x}=\bm{f}(\bm{\gamma},\bm{d},\hat{\mathcal{J}})
\end{array}\right.
\label{eq:load_coupling}
\end{equation}

For the load coupling, the composition of function $\bm{f}$ and $\bm{h}$, i.e. $\bm{f}(\bm{h}(\bm{x},\bm{p},\hat{\mathcal{I}}),\bm{d},\hat{\mathcal{J}})$, is a standard interference function (SIF) \cite{Yates:1995eh} in $\bm{x}$. The proof of this is in \cite{You:2015wva}. The definition of the SIF is given below \cite{Cavalcante:2014jd}. 

\begin{definition}
A function $\bm{f}$: $\mathbb{R}^m_+\rightarrow\mathbb{R}_{++}$ is called an SIF if the following properties hold:
\begin{enumerate}
\item (Scalability) $\alpha \bm{f}(\bm{x})>\bm{f}(\alpha\bm{x}),~\forall \bm{x}\in \mathbb{R}^m_+,~\alpha>1$.
\item (Monotonicity) $\bm{f}(\bm{x})\geq \bm{f}(\bm{x'})$, if $\bm{x} \geq \bm{x'}$.
\end{enumerate}
\label{def:SIF}
\end{definition}

In \cite{Yates:1995eh}, the convergence point of an SIF is proved to be unique, and can be obtained by the fix-point iteration. Suppose the association between cell and UE is fixed. For any given power vector $\bm{p}$ and demand vector $\bm{d}$, we have a unique load vector $\bm{x}$ satisfying $\bm{x}=\bm{f}(\bm{h}(\bm{x},\bm{p}),\bm{d})$. We call the relationship of the mutual influence between $\bm{p}$ and $\bm{x}$ the \textbf{power-load coupling}.
For the SIF, we introduce some propositions and lemmas. They work as fundamentals for the analysis on theoretical aspects for energy minimization. Proposition~\ref{lma:increasing} comes directly from the monotonicity of the SIF. The proof of Lemma~\ref{lma:1/x} and Lemma~\ref{lma:Ax} can be found in \cite{Anonymous:Dv6oz3oX}. Proposition~\ref{lma:increasing} and Proposition~\ref{lma:sigma} are used as the fundamental for proving Lemma~\ref{lma:x'>=x} and Lemma~\ref{lma:alpha_x} in Section~\ref{sec:energy}, respectively.

\begin{proposition}
For the sequence $\bm{x}^{(0)},\bm{x}^{(1)},\ldots$ generated by fix-point iteration, if there exists $k$ satisfying $\bm{f}(\bm{x}^{(k)}) \leq\bm{f}(\bm{x}^{(k-1)})$, then the sequence $\bm{x}^{(k)},\bm{x}^{(k+1)},\ldots$ is monotonously decreasing (in every component), otherwise if $\bm{f}(\bm{x}^{(k)})\geq\bm{f}(\bm{x}^{(k-1)})$, then the sequence is monotonously increasing (in every component).	
\label{lma:increasing}
\end{proposition}

\begin{lemma}
The function $\mathbb{R}\rightarrow\mathbb{R}:x\mapsto1/\log[1/(1+\frac{1}{x})]$ is concave.
\label{lma:1/x}
\end{lemma}

\begin{lemma}
Suppose $\bm{A}\in\mathbb{R}^{n\times m}$ and $\bm{b}\in\mathbb{R}^{n}$. Define $\varphi:\mathbb{R}^n_+\rightarrow\mathbb{R}$. If $\varphi$ is concave, so is $f(\bm{A}\bm{x}+\bm{b})$.
\label{lma:Ax}
\end{lemma}

\begin{proposition}
Scalability holds for $\bm{f}(\bm{h}(\sigma^2,\bm{x}))$ in $[\sigma^2,\bm{x}]$.	
\label{lma:sigma}
\end{proposition}
\begin{proof}
The function $\bm{f}(\bm{h}(\sigma^2,\bm{x}))$ can be obtained by making linear transformation for $x$ in the function $1/\log_2(1/(1+1/x))$, of which the concavity still holds after the transformation, according to Lemma~\ref{lma:1/x} and Lemma~\ref{lma:Ax}. Hence the conclusion.
\end{proof}


\section{Problem Formulation}

The energy minimization problem (\textit{MinE}), is formulated in this section. For the sake of mathematical presentation, the association, i.e. $\hat{\mathcal{I}}$ (or $\hat{\mathcal{J}}$), is replaced by an $n\times m$ matrix $\bm{\kappa}$. And $\kappa_{ij}=1$ means that cell $i$ is currently serving UE $j$. In other words, we have
	$\kappa_{ij}=1~~\Leftrightarrow~~i\in\mathcal{I}_j$~and~$j\in\mathcal{J}_i$, 
	$\kappa_{ij}=0~~\Leftrightarrow~~i\notin\mathcal{I}_j~\textnormal{and}~j\notin\mathcal{J}_i$, with
	$\kappa_{ij}\in\{0,1\}$.
The objective is to minimize the total energy consumption for transmission. 
Recall that for all $i\in\mathcal{I}$, $p_i$ is the transmit power per RU in cell $i$, and $x_i$ is the proportion of allocated RUs for transmission in cell $i$. 
Therefore, the total energy consumed in cell $i$ is $\bm{p}^{\mathsf{T}}\bm{x}=\sum_{i\in\mathcal{I}}p_i(x_i\times M)$, where $x_i\times M$ is the number of occupied RUs in cell $i$.
The formulation is as follows.
\begin{subequations}
\begin{alignat}{2}
[\textit{MinE}]&\quad\min\limits_{\bm{\kappa},\bm{d},\bm{p},\bm{x}} \quad \bm{p}^{\mathsf{T}}\bm{x}\\
 \textnormal{s.t.} &\quad  \bm{x}=\bm{f}(\bm{h}(\bm{x},\bm{p},\bm{\kappa}),\bm{d},\bm{\kappa})  \\
   & \quad \bm{p}\leq \bm{p}_{max} \\
    & \quad \bm{d}\geq\bm{d}_{min} \\
 &\quad 0< x_{i}\leq 1 &\forall i\in\mathcal{I} \\
 & \quad \kappa_{ij}\in\{0,1\} &\forall i\in\mathcal{I},~j\in\mathcal{J}
\end{alignat}
\label{eq:p1}
\end{subequations}
 Due to that $M$ is a constant, we set the objective of \textit{MinE} to be $\sum_{i\in\mathcal{I}}p_ix_i$, shown in (\ref{eq:p1}a) in \textit{MinE}. 
We impose (\ref{eq:p1}b) to be the power-load coupling constraint. The inequalities (\ref{eq:p1}c) and (\ref{eq:p1}d) are constraints for the maximal transmit power and the minimal user demand, respectively. In constraint (\ref{eq:p1}e), the cell load is limited to be no more than $\bm{x}=\bm{1}$, aka. the full load. The cell-UE association is one of the variables in \textit{MinE}, imposed to be binary in (\ref{eq:p1}f). The variables in \textit{MinE} are $\bm{\kappa}$, $\bm{d}$, $\bm{x}$, and $\bm{p}$.

\section{Energy Minimization: Analysis and Solution}
\label{sec:energy}

In this section, we give theoretical analysis on how to optimize the transmission energy. Based on the theoretical properties, we respectively propose a power allocation algorithm and an algorithm to optimize the cell-UE association. 

\begin{proposition}
$\bm{d}=\bm{d}^{min}$ is optimal for \textit{MinE}.
\label{lma:d_min}
\end{proposition}
\begin{proof}
According to Proposition~\ref{lma:increasing}, if we reduce the data rate, then the load at convergence will decrease. Thus the optimal setting of demand is $\bm{d}=\bm{d}^{min}$.
\end{proof}

According to Proposition~\ref{lma:d_min}, we can set $\bm{d}$ to $\bm{d}^{min}$ in \textit{MinE} without loss of optimality. For the clarity of discussion, we set $\bm{d}=\bm{d}^{min}$. Then the variables in this section are the power $\bm{p}$, the load $\bm{x}$ and the association $\bm{\kappa}$. For the sake of presentation, we use $\textnormal{\textit{Fix}}\{\bm{f}(\bm{h}(\bm{x},\bm{p},\bm{\kappa}),\bm{\kappa})\}$ to denote the fixed point of the function $\bm{f}(\bm{h}(\bm{x},\bm{p},\bm{\kappa}),\bm{\kappa})$. We formally define the notations of $\bm{p}$, $\bm{p'}$, $\bm{x}$ and $\bm{x'}$, in Definition~\ref{def:p_x} below. 

\subsection{Analysis on Power Allocation}
In \cite{Ho:2015hw,Anonymous:SLcg4Bln}, the authors showed the optimality of the full load in the non-JT case. However, in the JT case, this conclusion does not hold, because the full load may not be possible for all cells. For example, suppose there are two cells $c_1$ and $c_2$, and two UEs $u_1$ and $u_2$. Cell $c_1$ serves $u_1$ and $u_2$, and cell $c_2$ serves only $u_2$. If the demands of both $u_1$ and $u_2$ are non-zero, then $c_2$ cannot be in full load. Even we know the optimal load setting for \textit{MinE}, the power solution is not always unique. Due to these reasons, the conclusion in \cite{Ho:2015hw,Anonymous:SLcg4Bln} does not hold in general with JT. In this subsection, we show that, in but not limited to the JT case, the energy performance always benefits from letting the power of all cells be scaled down uniformly.

\begin{definition} Given any association $\bm{\kappa}$, we define the notations $\bm{p}$, $\bm{p'}$, $\bm{x}$ and $\bm{x'}$ as follows.
\begin{enumerate}
	\item Denote $\bm{p'}=\bm{p}/\alpha$, with $\alpha>1$.
	\item Denote $\bm{x}=\bm{f}(\bm{h}(\bm{x},\bm{p}))$, i.e. $\bm{x}$ is the fixed point corresponding to $\bm{p}$. 
	\item Denote $\bm{x'}=\bm{f}(\bm{h}(\bm{x'},\bm{p'}))$, i.e. $\bm{x'}$ is the fixed point corresponding to $\bm{p'}$. 
\end{enumerate}
	\label{def:p_x}
\end{definition}

In Definition~\ref{def:p_x}, $\bm{p'}$ is the power scaled down from $\bm{p}$ by the scaling constant $\alpha>1$. The load $\bm{x'}$ is the fixed point in the power-load coupling equations with the power $\bm{p'}$. Lemma~\ref{lma:x'>=x} and Lemma~\ref{lma:alpha_x} show properties for such a scaling operation on $\bm{p}$. Lemma~\ref{lma:x'>=x} shows that, when the power is scaled down, the load increases. Lemma~\ref{lma:alpha_x} gives an upper bound on the increased load $\bm{x'}$. Based on Lemma~\ref{lma:alpha_x}, we get Theorem~\ref{thm:pre_opt} that serves as a theoretical support for the proposed power allocation algorithm. Based on Lemma~\ref{lma:x'>=x} and Theorem~\ref{thm:pre_opt}, we propose our power allocation algorithm.

\begin{lemma}
$\bm{x'}\geq\bm{x}$.
\label{lma:x'>=x}
\end{lemma}
\begin{proof}
We fix $\bm{\kappa}$ in \textit{MinE} in the proof. 
Let $\bm{x}^{(1)}=\bm{f}(\bm{h}(\bm{p'},\bm{x}^{(0)}))$ with $\bm{x}^{(0)}=\bm{x}$. According to Eq.~(\ref{eq:sinr}), we can verify that $\bm{f}(\bm{h}(\bm{p}/\alpha,\bm{x}))=\bm{f}(\bm{h}(\sigma^2\alpha,\bm{x},\bm{p}))$. Because $\bm{f}(\bm{h}(\sigma^2\alpha,\bm{x},\bm{p}))$ is increasing in $\sigma^2$, we have $\bm{x}^{(1)}\geq\bm{x}^{(0)}=\bm{x}$. By Proposition~\ref{lma:increasing}, we have $\bm{x'}\geq\bm{x}$. Hence the conclusion.
\end{proof}

\begin{lemma}
$\bm{x'}<\alpha\bm{x}$.
\label{lma:alpha_x}
\end{lemma}
\begin{proof}
We fix $\bm{\kappa}$ in \textit{MinE} in the proof. 
By Lemma~\ref{lma:sigma}, $\bm{f}(\bm{h}(\sigma^2,\bm{x}))$ is concave in $\sigma^2$, and thus the scalability holds for $\sigma^2$. 
Denote by $\bm{x}^{(1)}$ the first iteration round, with $\bm{x}^{(1)}=\bm{f}(\bm{h}(\bm{p}/\alpha,\bm{x}^{(0)}))$ and $\bm{x}^{(0)}=\bm{x}$.
According to Lemma~\ref{lma:sigma}, we have $\bm{x}^{(1)}<\alpha\bm{x}$.
Then the remaining proof can be established by mathematical induction. Suppose for some $k\geq 1$, we have 
$
\bm{x}^{(k)}<\alpha\bm{x}	
$. 
By the monotonicity of $\bm{f}(\bm{h}(\alpha\sigma^2,\bm{x},\bm{p}))$ in $\bm{x}$, and Lemma~\ref{lma:Ax}, we have 
$\bm{x}^{(k+1)}\leq\bm{x}^{(k)}<\alpha\bm{x}$. 
Because we have $\bm{x}^{(k)}<\alpha\bm{x}$ for $k=1$, we can conclude that $\bm{x}^{(k)}<\alpha\bm{x}$ holds for any $k\geq 2$, until the convergence $\bm{x'}=\bm{f(\bm{h}(\bm{p'},\bm{x'}))}$ is reached. Hence the conclusion.
\end{proof}

\begin{theorem}
$\bm{p'}^{\mathsf{T}}\bm{x'}\leq\bm{p}^{\mathsf{T}}\bm{x}$.
\label{thm:pre_opt}
\end{theorem}
\begin{proof}
According to Lemma~\ref{lma:alpha_x}, we have $\bm{x'}<\alpha\bm{x}$. So
\begin{equation}
\bm{p'}^{\mathsf{T}}\bm{x'}=\frac{\bm{p}^{\mathsf{T}}}{\alpha}\bm{x'}<\frac{\bm{p}^{\mathsf{T}}}{\alpha}\cdot\alpha\bm{x}=\bm{p}^{\mathsf{T}}\bm{x}\textnormal{,}
\end{equation}
hence the conclusion.
\end{proof}

Based on Lemma~\ref{lma:x'>=x} and Theorem~\ref{thm:pre_opt}, we design a bisection search based algorithm, namely, POwer scaLed dOwn, (\textit{POLO}), to compute a power allocation that improves the energy performance. \textit{POLO} is shown in Algorithm~\ref{alg:POLO}. The input of \textit{POLO} is the power $\bm{p}$, the load $\bm{x}$ and the association $\bm{\kappa}$. The output is the optimized power allocation $\bm{p'}$ and its corresponding load vector $\bm{x'}$. Let us see how \textit{POLO} works. The basic idea of \textit{POLO} is to find a scaling constant $\beta~(0<\beta<1)$, that makes the new power $\bm{p'}$ to be less than the original power $\bm{p}$, while not violating the maximum load constraint $\bm{x}\leq\bm{1}$. By Lemma~\ref{lma:x'>=x}, we know that the load will increase from $\bm{x}$ to $\bm{x'}$, when we reduce the power from $\bm{p}$ to $\bm{p'}$. On the other hand, the total energy consumption will decrease with $\bm{p'}$ and $\bm{x'}$, according to Theorem~\ref{thm:pre_opt}. As shown in Line 4, $\bm{x}$ is updated to $\bm{x'}$ in each iteration round. When the norm of $\bm{p'}-\bm{p''}$ is no more than a given small value $\epsilon$, it means that $\beta$ is (in respect of $\epsilon$) at convergence to the maximum value that makes $\bm{x'}$ satisfying the full load constraint (\ref{eq:p1}e). In this case, \textit{POLO} ends and returns the new power allocation $\bm{p'}$ and the corresponding load $\bm{x'}$. 

\begin{algorithm}[t]
\label{alg:POLO}
\caption{Power Allocation (\textit{POLO})}
\KwIn{$\bm{p}$, $\bm{x}$, $\bm{\kappa}$, $\epsilon$}
\KwOut{$\langle \bm{p'},\bm{x'}\rangle$}
$\check{\beta}\leftarrow 0$, $\hat{\beta}\leftarrow 1$, $\bm{p'}\leftarrow\bm{p}+\bm{\epsilon}$, $\bm{p''}\leftarrow\bm{p}$\;
\While{$||\bm{p'}-\bm{p''}||>\epsilon$}
{
	$\beta\leftarrow(\check{\beta}+\hat{\beta})/2$\;
	$\bm{p''}\leftarrow\bm{p'}$\;
	$\bm{p'}\leftarrow\beta\bm{p}$\;
	$\bm{x'}\leftarrow \textnormal{\textit{Fix}}\{\bm{f}(\bm{h}(\bm{x},\bm{p'},\bm{\kappa}),\bm{\kappa})\}$\;
	$\bm{x}\leftarrow\bm{x'}$\;
	\eIf{$\max_{i\in\mathcal{I}}x'_i>1$}{
		$\check{\beta}\leftarrow\beta$\;
	}{
		$\hat{\beta}\leftarrow\beta$\;
	}
}
\end{algorithm}

\subsection{Analysis on Association Optimization}
With the power allocation algorithm \textit{POLO}, the total energy decreases with the increase of the network load $\bm{x}$. If there exists a cell $i$ with $x_i=1$, we cannot find a scaling constant $0<\beta<1$ to further reduce the power. In this subsection, we propose an algorithm to optimize the association that reduces the network load via JT. We  use $c$ and $u$ to denote any cell and any UE in the network, respectively, to avoid conflicting with the index $i$ and $j$. Before we change the association, the load and SINR functions are denoted by $f_c(\bm{x})$ and $h_u(\bm{\gamma})$ for any $c\in\mathcal{I}$ and $u\in\mathcal{J}$, respectively.
 
Consider the case of adding a link from a cell $c$ to a UE $u$. Suppose UE $u$ is currently served by the cell(s) in $\mathcal{I}_u$ and there is some cell $c\notin\mathcal{I}_u$. Suppose the set of UEs served by cell $c$ is $\mathcal{J}_c$. The set of UE $u$'s serving cells expands from $\mathcal{I}_u$ to $\mathcal{I}_u\cup\{c\}$. Then UE $u$ does not receive interference from cell $c$ anymore, so the set of cells generating interference to $u$ contracts from $\mathcal{I}\backslash\mathcal{I}_j$ to $\mathcal{I}\backslash(\mathcal{I}_j\cup\{c\})$. The SINR function of UE $u$ after making $c$ to serve $u$, is denoted by $h_u^{+}$, shown in Eq.~(\ref{eq:adding_vartheta}). Note that for any $j\neq u$, the association between UE $j$ and its serving cells $\mathcal{I}_j$ does not change. That is, we have $h^+_j(\bm{x})=h_j(\bm{x})$, for all $j\neq u$. Denote $\bm{h}^+(\bm{x})=[h^+_1(\bm{x}),h^+_2(\bm{x}),\ldots,h^+_m(\bm{x})]$.

\begin{equation}
{h}^{+}_u(\bm{x})=
\frac{\sum_{i\in \mathcal{I}_u\bigcup\{c\}}p_ig_{iu}}{\sum_{k\in \mathcal{I}\backslash(\mathcal{I}_j\cup\{c\})}p_kg_{ku}x_{ku}+\sigma^2}
\label{eq:adding_vartheta}
\end{equation}

For cell $c$, the set of its serving UE is expanded from $\mathcal{J}_c$ to $\mathcal{J}_c\cup\{u\}$. We formulate the load function of cell $c$, represented by notation $f_c^{+}$, in Eq.~(\ref{eq:adding_varphi}). Note that for any $i\neq c$, the association between cell $i$ and its served UEs $\mathcal{J}_i$ does not change. In other words, we have $f^+_i(\bm{x})=f_i(\bm{x})$, for all $i\neq c$. Denote $\bm{f}^+(\bm{\gamma})=[f^+_1(\bm{\gamma}),f^+_2(\bm{\gamma}),\ldots,f^+_n(\bm{\gamma})]$. 

\begin{equation}
{f}^{+}_c(\bm{\gamma})=\sum_{j\in\mathcal{J}_c\bigcup\{u\}}\frac{r_j}{MB\log_2\left(1+\gamma_u\right)}
\label{eq:adding_varphi}
\end{equation}

We show the following theorem, given as a sufficient condition to judge if adding the downlink from $c$ to $u$ can reduce the network load $\bm{x}$. The notation $\circ$ used below denotes the function compound relationship \cite{Anonymous:Dv6oz3oX}. That is, $\bm{f}\circ\bm{h}(\cdot)$ is equivalent to $\bm{f}(\bm{h}(\cdot))$. 

\begin{theorem}
Suppose $\bm{\widetilde{x}}=\bm{{f}}\circ\bm{{h}}(\bm{\widetilde{x}})$ and $\bm{x}=\bm{{f}^{+}}\circ\bm{{h}^{+}}(\bm{x})$. Then $\bm{x}\leq\bm{\widetilde{x}}$ if $\exists k\geq 1$ in the iteration $\bm{x}^{(k)}=\bm{{f}}\circ\bm{{h}}^{+}(\bm{x}^{(k-1)})$ such that ${f}^{+}_c\circ\bm{{h}}^{+}(\bm{x}^{(k)})\leq x^{(k)}_c$, where $\bm{x}^{(0)}=\bm{\widetilde{x}}$.
\label{thm:adding_sufficient}
\end{theorem}
\begin{proof} 
We construct the following steps.
In the first $k$ iteration rounds (iteration $t\in[1,k]$),
let $\bm{x}^{(t)}=\bm{{f}}\circ\bm{{h}}^{+}(\bm{x}^{(t-1)})$. By Eq.~(\ref{eq:adding_vartheta}), we have
$
\bm{x}^{(1)}=\bm{{f}}\circ\bm{{h}}^{+}(\bm{x}^{(0)})\leq\bm{{f}}\circ\bm{{h}}(\bm{x}^{(0)})=\bm{x}^{(0)}
$.
Thus by Lemma~\ref{lma:increasing}, $\bm{x}^{(k)}\leq\bm{x}^{(k-1)}\leq\cdots\leq\bm{x}^{(0)}$ holds.
For the remaining iteration rounds (iteration $t>k$), let $\bm{x}^{(k+1)}=\bm{{f}}^{+}\circ\bm{{h}}^{+}(\bm{x}^{(k)})$. According to the condition in Theorem~\ref{thm:adding_sufficient}, ${f}^{+}_c\circ\bm{{h}}^{+}(\bm{x}^{(k)})\leq x^{(k)}_c$ holds. For any $i\neq c$, by Eq.~(\ref{eq:adding_varphi}), we have 
${f}^{+}_i\circ\bm{{h}}^{+}(\bm{x}^{(k)})={f}_i\circ\bm{{h}}^{+}(\bm{x}^{(k)})$. 
Since $\bm{x}^{(k-1)}\geq\bm{x}^{(k)}$, we have
$
{f}_i\circ\bm{{h}}^{+}(\bm{x}^{(k)})\leq{f}_i\circ\bm{{h}}^{+}(\bm{x}^{(k-1)})=\bm{x}^{(k)}
$
holds. Then we have 
$
\bm{x}^{(k+1)}=\bm{{f}}^{+}\circ\bm{{h}}^{+}(\bm{x}^{(k)})\leq\bm{x}^{(k)}
$. 
By Lemma~\ref{lma:increasing}, at convergence we have
$
\bm{x}=\bm{{f}}^{+}\circ\bm{{h}}^{+}(\bm{x})\leq\cdots\leq\bm{x}^{(k+1)}\leq\bm{x}^{(k)}
\label{eq:x_converge}	
$. 
Combined with $\bm{x}^{(k)}\leq\bm{x}^{(0)}$, we get the conclusion.
\end{proof}


\begin{algorithm}[t]
\label{alg:AOLO}
\caption{Association Allocation (\textit{AOLO})}
\KwIn{$\bm{p}$, $\bm{x}$, $\bm{\kappa}$, $\tau$}
\KwOut{$\langle\bm{\kappa'},\bm{x'}\rangle$}
\For{$\forall \kappa_{ij}=0~\textnormal{in}~\bm{\kappa}$}
{
		$\bm{x}^{(0)}\leftarrow\bm{x}$\;
		\For{$k\leftarrow1\textnormal{ \textbf{to} }\tau$}
		{
			\If{$\kappa_{ij}=0$}
			{
				$\kappa'_{ij}\leftarrow 1$\;
				$\bm{x}^{(k)}\leftarrow\bm{f}\left(\bm{h}(\bm{x}^{(k-1)},\bm{p},\bm{\kappa'}),\bm{\kappa}\right)$\;
				\If{$x_i^{(k)}\leq f_i\left(\bm{h}(\bm{x}^{(k)},\bm{p},\bm{\kappa'}),\bm{\kappa'}\right)$}
				{
					$\kappa_{ij}\leftarrow\kappa'_{ij}$\;
						$\bm{x'}\leftarrow \textnormal{\textit{Fix}}\{\bm{f}(\bm{h}(\bm{x},\bm{p},\bm{\kappa'}),\bm{\kappa'})\}$\;
				}
			}
		}
}
\end{algorithm}

\begin{algorithm}[b]
\label{alg:PALO}
\caption{Power-Association Allocation (\textit{PALO})}
\KwIn{$\bm{p}$, $\bm{x}$, $\bm{\kappa}$, $\epsilon$, $\tau$}
\KwOut{$\bm{p'}$, $\bm{\kappa'}$}
\Repeat{$\bm{\kappa}=\bm{\kappa'}$}
{
$\langle\bm{p'},\bm{x'}\rangle\leftarrow$\textit{ POLO}$(\bm{p},\bm{x},\bm{\kappa},\epsilon)$\;
$\langle\bm{\kappa'},\bm{x''}\rangle\leftarrow$\textit{ AOLO}$(\bm{p'},\bm{x'},\bm{\kappa},\tau)$\;
$\langle\bm{p},\bm{x},\bm{\kappa}\rangle\leftarrow\langle\bm{p'},\bm{x''},\bm{\kappa'}\rangle$\;
}
\end{algorithm}

Based on Theorem~\ref{thm:adding_sufficient}, an algorithm for optimizing AssOciation for reducing the LOad, namely \textit{AOLO}, is proposed. \textit{AOLO} is shown in Algorithm~\ref{alg:AOLO}. The basic idea of \textit{AOLO}, is to check all cell-UE pairs without a downlink, to see if adding a link from the cell to the UE can reduce the network load $\bm{x}$, by the derived sufficient condition in Theorem~\ref{thm:adding_sufficient}. Since the transmit power $\bm{p}$ is not changed in \textit{AOLO}, the energy consumption, $\bm{p}^{\mathsf{T}}\bm{x}$, is reduced. In \textit{AOLO}, the parameter $\tau$ is a given constant, indicating the maximal number of iteration rounds to check if the condition in Theorem~\ref{thm:adding_sufficient} can be satisfied. \textit{AOLO} checks all elements $\kappa_{ij}$ with $\kappa_{ij}=0$ in $\bm{\kappa}$. As shown in Line 9, the fix-point iteration is executed only if the new association $\bm{\kappa'}$ is ascertained to lead to a reduction on the network load $\bm{x}$. \textit{AOLO} ends when no improvement on the load can be found by the sufficient condition in Theorem~\ref{thm:adding_sufficient}. In this case, the new association $\bm{\kappa'}$ and its corresponding network load $\bm{x'}$ are returned.

\section{Algorithm Design and Analysis}

Based on the two algorithms \textit{POLO} and \textit{AOLO}, we propose an algorithm that jointly optimizes the power $\bm{p}$ and the association $\bm{\kappa}$. The proposed algorithm is named \textit{PALO} and shown in Algorithm~\ref{alg:PALO}. The basic idea of \textit{PALO} is to improve the power allocation and association allocation by \textit{POLO} and \textit{AOLO}, iteratively. In every iteration round, the old tuple of power, association, and load, $\langle\bm{p},\bm{\kappa},\bm{x}\rangle$, is updated to the new tuple $\langle\bm{p'},\bm{x'},\bm{\kappa'}\rangle$. In Line 3, we get a new tuple of power and load, $\langle\bm{p'},\bm{x'}\rangle$, with $\bm{p'}<\bm{p}$ and $\bm{p}^{\mathsf{T}}\bm{x}>\bm{p'}^{\mathsf{T}}\bm{x'}$. In Line 4, under the power allocation $\bm{p'}$, we further seek for a new tuple of association and load, $\langle\bm{\kappa'},\bm{x''}\rangle$, with $\bm{x'}\geq\bm{x''}$. Then, in each round of the loop shown in Lines 1--5, we have the inequality $\bm{p}^{\mathsf{T}}\bm{x}>\bm{p'}^{\mathsf{T}}\bm{x'}\geq\bm{p'}^{\mathsf{T}}\bm{x''}$. Therefore, the energy is improved round by round, until we cannot reduce the network load by adding JT link. In this case, \textit{AOLO} cannot find a new association that reduce the load, so we have $\bm{\kappa'}=\bm{\kappa}$. Besides, since the $\beta$ obtained from \textit{POLO} is now the maximal value that does not make the load vector to violate the full load constraint, $\beta$ cannot be updated again. Therefore \textit{PALO} ends. Theorem~\ref{thm:polynomial} shows that \textit{PALO} is in polynomial time.
\begin{theorem}
Suppose that computing the fixed point of $\bm{f}(\bm{h}(\bm{x}))$ is of the complexity $O(K)$, then \textit{PALO} runs in $O(Km^2n^2)$.
\label{thm:polynomial}	
\end{theorem}
\begin{proof}
In \textit{POLO}, the maximal possible distance between $\bm{p'}$ and $\bm{p''}$ is less than that between $\bm{p}^{max}$ and $\bm{0}$, due to the assumption that we have $\bm{p}>\bm{0}$ for any $\bm{p}$. Consider the worst case, that the initial power is $\bm{p}^{max}$ and the finally optimized power is a very small positive value. Let $p^{max}$ be the largest value in $\bm{p}^{max}$, we have Eq.~(\ref{eq:complexity1}) holds.
\begin{equation}
\sqrt{(p'_1-p''_1)^2+\cdots+(p'_n-p''_n)^2}<\sqrt{n\beta^2\left(p^{max}\right)^2}
\label{eq:complexity1}
\end{equation}
We get $||\bm{p'}-\bm{p''}||<\epsilon$ if and only if $\beta\leq\epsilon/(\sqrt{n}p^{max})$. This can be verified by replace any such $\beta$ in Eq.~(\ref{eq:complexity1}), i.e.
\begin{equation}
||\bm{p'}-\bm{p''}||<\sqrt{n\left(\frac{\epsilon}{\sqrt{n}p^{max}}\right)^2\left(p^{max}\right)^2}=\epsilon
\end{equation}
Thus, to get $\beta\leq\epsilon/(\sqrt{n}p^{max})$ by reducing $\beta$ via the bisection search, the complexity is 
$
O(\log_2\left[\left(\sqrt{n}p^{max})/\epsilon\right]\right)=O(\log_2n)
$. 
So \textit{POLO} runs in $O(K\log_2n)$. 
In \textit{AOLO}, the outer loop repeats at most $m\times n$ round. So \textit{AOLO} is of the complexity $O(Kmn)$. In PALO, note that the loop repeats at most $m\times n$ round. Then the computation complexity of \textit{PALO} is calculated as
\begin{equation}
O(mn)\times\left[O(K\log_2n)+O(Kmn)\right]=O(Km^2n^2)
\end{equation}
Hence the conclusion.
\end{proof}

\section{Performance Evaluation}

\subsection{Simulation Settings}


We show numerical results in this section. In the simulation, there are 7 macro cells (MCs) in total. Two small cells (SCs) are placed around each MC. Thirty UEs are randomly distributed in each hexagon. The network operates at 2 GHz. Each RU is set to 180 kHz bandwidth and the bandwidth for each cell is 4.5 MHz. The noise power spectral density is set to -174 dBm/Hz. The path loss for the MCs follows the standard 3GPP urban macro (UMa) model \cite{Anonymous:GtzCaaVU}. The path loss for the SCs follows the standard 3GPP urban micro (UMi) model of hexagonal deployment. The shadowing coefficients are generated by the log-normal distribution with 6 dB and 3 dB standard deviation \cite{Anonymous:GtzCaaVU}, for MCs and SCs, repectively. The maximum transmit power levels for MCs and SCs are set to 200 mW and 50 mW per RU, respectively. Each UE is initially connected to the cell (an MC or SC) with the best received signal power, i.e., the network is initialized without JT. We run the simulation on 15 groups of data. The final results shown in this section are averaged from them. 

\subsection{Numerical Results}

\begin{figure}[t]
\centering
  \includegraphics[width=0.927\linewidth]{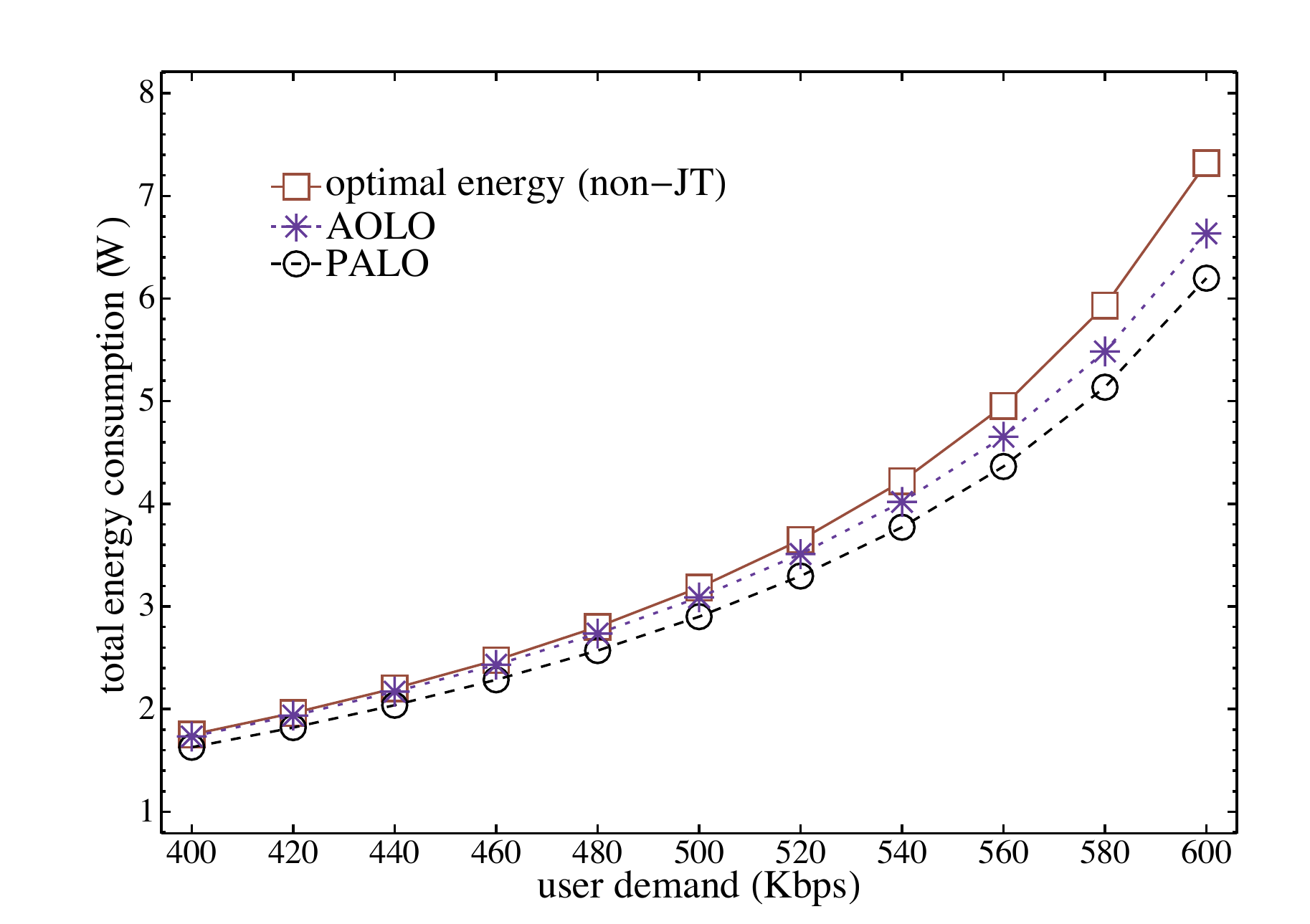}
  \vskip -5pt
  \caption{Optimal non-JT vs. Sub-optimal JT}
  \label{fig:plot1}
\end{figure}

\figurename~\ref{fig:plot1} shows the results of total energy consumption among three cases. In the case of optimal energy in non-JT, each UE is served only by one cell (the cell with the best received signal power). Due to the full load optimality in the non-JT case \cite{Ho:2015hw}, each cell in this case is set to full load. The power of the non-JT case is then computed under the full load. In the case AOLO shown in \figurename~\ref{fig:plot1}, only the algorithm \textit{AOLO} is applied. In other words, the transmit power is as same as that in the non-JT case, while the cell load is reduced by making some cells to serve UEs via JT. In the case PALO, the algorithm \textit{PALO} is applied on the non-JT case. On average, the optimal energy in non-JT can be reduced by $4.01\%$ via optimizing the association by \textit{AOLO}. By iteratively using \textit{AOLO} and \textit{POLO}, i.e. \textit{PALO}, the energy can be further reduced by $6.71\%$. Compared with the initial non-JT case, the energy is reduced by $9.82\%$ by \textit{PALO}. The improvement becomes larger with the increase of the user demand.

In \figurename\ \ref{fig:plot2}, we investigate the performance of the power allocation algorithm \textit{POLO} under JT. The cell-UE association for each user demand is computed by \textit{PALO}. In CASE 1, we set the transmit power for MCs and SCs per RU to 160 mW and 40 mW, respectively. In CASE 2, we respectively set 120 mW and 30 mW for MCs and SCs. We then apply \textit{POLO} on both of the two cases. As we can see, the power allocation algorithm \textit{POLO} significantly reduces the energy consumption for both CASE 1 and CASE 2. Numerically, \textit{POLO} reduces the energy by $54.90\%$ and $45.28\%$ for CASE 1 and CASE 2, respectively. One can observe that the lower the user demand is, the more the energy consumption is reduced. This is because, when the user demand is low, the maximal cell load level is relatively low. Thus, the power can be reduced more with the cell being not overloaded, thus leading to more reduction on the energy. When the demand is 540 Kbps, there exists one cell at full load in the simulation, for both CASE 1 and CASE 2. The power cannot be reduced by \textit{POLO} in this situation, otherwise there will be at least one cell $i$ with $x_i>1$, violating the load constraint. For other demands less than 540 Kbps, there is visually no difference between the performance of \text{POLO} applied on CASE 1 and CASE 2. One can observe that \textit{POLO} is effective in energy saving in the JT case, when the cell-UE association is fixed.

\begin{figure}[t]
\centering
  \includegraphics[width=0.927\linewidth]{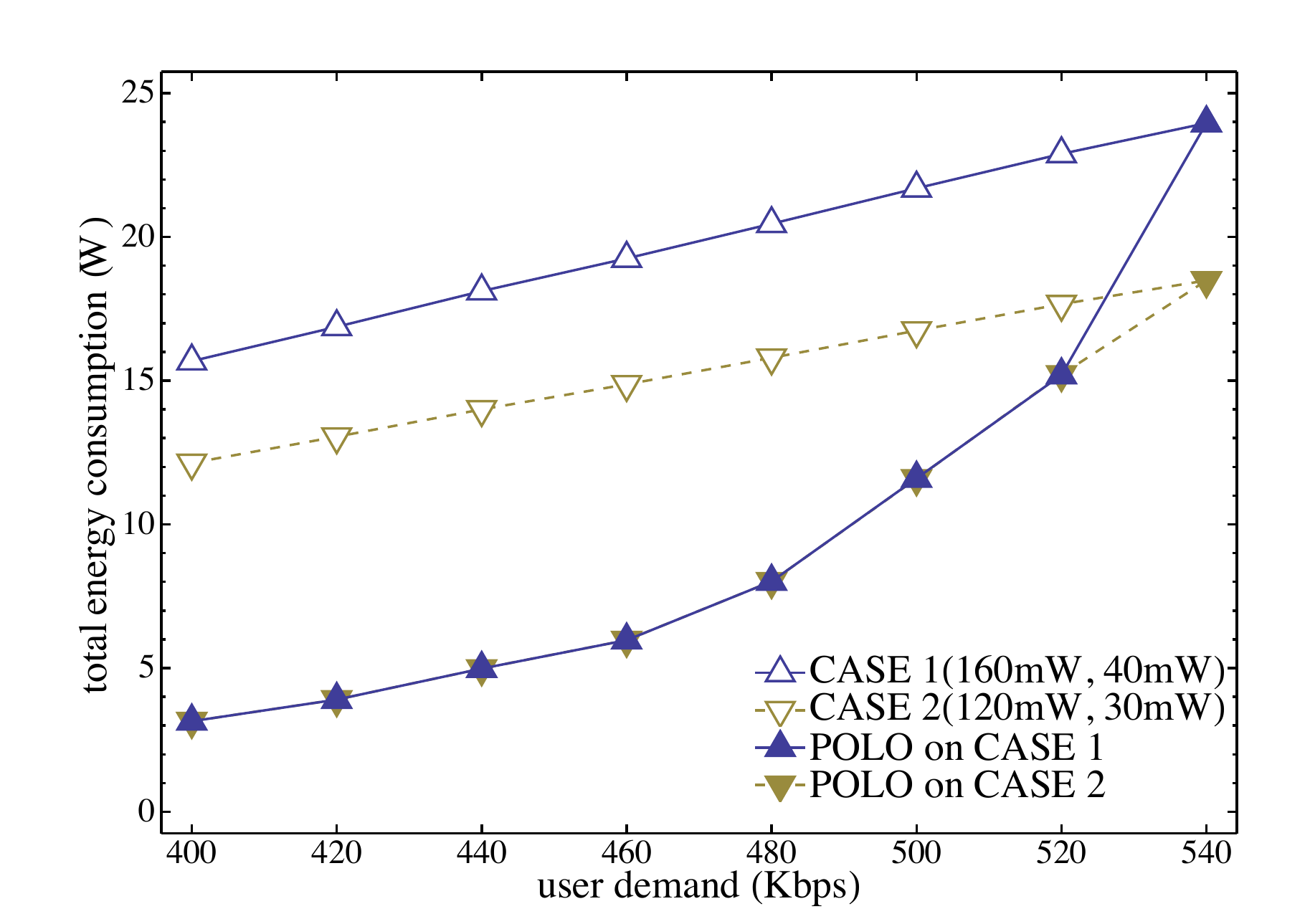}
  \caption{Performance of POLO under JT}
  \label{fig:plot2}
\end{figure}

\section{Conclusion}

We have investigated the energy minimization problem under the load-coupling model in JT scenarios. 
We remark that the conclusion of full load optimality in the previous work \cite{Ho:2015hw,Anonymous:SLcg4Bln} does not hold with JT in general. Thus, the scheme of optimal power and load setting in non-JT does not apply in the JT cases. For energy saving in JT, we consider to jointly optimize the power allocation and cell-UE association. Two algorithms are then proposed. We show theoretically that these two algorithms naturally combined with each other for energy saving.
We numerically showed that the proposed algorithms improve the energy performance in comparison with the optimal energy setting in non-JT. In JT, the proposed power allocation algorithm significantly improves the performance of total energy consumption, compared to those with the setting of fixed power levels. 
\section*{Acknowledgements}
This work has been supported by the Swedish Research Council and the
Link\"{o}ping-Lund Excellence Center in Information Technology (ELLIIT), Sweden,
and the European Union Marie Curie project MESH-WISE (FP7-PEOPLE-2012-IAPP\@: 324515), DECADE (H2020-MSCA-2014-RISE: 645705), and WINDOW (FP7-MSCA-2012-RISE: 318992). The work of the second author has been partially supported by the China
Scholarship Council (CSC). The work of D. Yuan has been carried out within
European FP7 Marie Curie IOF project 329313.

\end{document}